\documentclass[12pt]{iopart}
\usepackage[english]{babel}  

\usepackage{multirow}
\usepackage{braket}
\usepackage{float}
\usepackage{graphicx}
\usepackage{dsfont}
\usepackage{iopams}
\usepackage{amstext}
\usepackage{amsthm}

\usepackage{wrapfig}
\usepackage{braket}
\usepackage[colorlinks=true,citecolor=blue,linkcolor=blue]{hyperref}

\newcommand{\ketbra}[2]{\ket{#1}\!\bra{#2}}
\newcommand{\freeset}{\mathcal{F}}
\newcommand{\allset}{\mathcal{S}}

\newcommand{\trace}{\mathop{\mathrm{Tr}}}
\newcommand{\diag}{\mathop{\mathrm{diag}}}
\newcommand{\one}{\mathds{1}}
\newcommand{\implies}{\Rightarrow}
\newcommand{\sfrac}[2]{#1/#2}

\newtheorem{Thm}{Theorem}
\newtheorem{Def}{Definition}

\newtheorem{Crl}[Thm]{Corollary}
\newtheorem{Lem}[Thm]{Lemma}

\newtheorem{counterexample}{Counterexample}

\begin{document}

\title{Continuity of  robustness measures in quantum resource theories}

\author{Jonathan Schluck, Gl\'{a}ucia Murta, Hermann Kampermann, Dagmar Bru{\ss} and Nikolai Wyderka}

\address{Institut f\"ur Theoretische Physik III, Heinrich-Heine-Universit\"at D\"usseldorf, Universit\"atsstr.~1, D-40225 D\"usseldorf, Germany}

\ead{wyderka@hhu.de}

\begin{abstract}
Robustness measures are increasingly prominent resource quantifiers that have been introduced for quantum resource theories such as entanglement and coherence. Despite the generality of these measures, their usefulness is hindered by the fact that some of their mathematical properties remain unclear, especially when the set of resource-free states is non-convex.
In this paper, we investigate continuity properties of different robustness functions. We show that their continuity depends on the shape of the set of free states. In particular, we demonstrate that in many cases, star-convexity is sufficient for Lipschitz-continuity of the robustness, and we provide specific examples of sets leading to non-continuous measures. 
Finally, we illustrate the applicability of our results by defining a robustness of teleportability and of quantum discord.

\end{abstract}

\maketitle

\section{Introduction}\label{intro}
Quantum systems exhibit ubiquitous non-classical features like entanglement, coherence and discord, which constitute resources for quantum information processing \cite{horodecki2009quantum, winter2016operational, streltsov2014quantum}. Over the recent years, the general theoretical framework of quantum resource theories has been developed \cite{chitambar2019quantum}. A quantum resource theory is defined by the set of its free elements (which do not exhibit the resource), the set of free operations (which do not increase the resource), and a quantifier of the resource. In order to quantify the amount of the resource, various different measures are used. A broadly applicable family of these quantifiers are robustness measures, which quantify the amount of noise tolerable before an element looses the property at hand \cite{RobOrig, GenRob}.
A part of the success of robustness measures lies in the fact that they have been linked to advantages of resource-carrying over resource-free elements for specific tasks like phase estimation \cite{napoli2016robustness} and channel discrimination \cite{takagi2019operational}. Recently, it has been shown that for each convex resource theory a discrimination task can be defined where the robustness exactly quantifies the advantage \cite{takagi2019operational}. Here, a quantum resource theory is referred to as convex, if its set of free elements is convex.

While robustness measures can be defined for any kind of resource, in the literature, they are mostly used for convex resource theories like entanglement, coherence and imaginarity. One reason for not using them in non-convex settings lies in the fact that even basic properties like continuity of the robustness measure are then not guaranteed anymore. This is in contrast to robustnesses defined w.r.t.~convex sets, where their continuity follows directly from convexity \cite{ConvexOpt}. However, this fact seems to be not widely known, which is probably the reason why we find discussions \cite{unique,Dagmar,Lami21} and several proofs for continuity of specific robustness measures in the literature \cite{unique, smooth, zheng2018generalized}, including a publication coauthored by one of the authors of the present work \cite{simnacher2019certifying}. 

The purpose of this paper is to investigate continuity features of robustness measures defined w.r.t.~different sets of resource-free elements. First, we introduce two different kinds of robustness measures used in the literature. We then review known continuity results for convex functions and improve them by showing that in many cases, these robustnesses are even Lipschitz-continuous. This stronger form of continuity can be used to obtain bounds on the difference of robustnesses of two elements from sole knowledge of their distances.

We then relax the condition of convexity to star-convexity and formulate conditions when these sets still lead to (Lipschitz-)continuous robustness measures.  Star-convexity is present in the resource-free sets of resource theories like discord \cite{ollivier2001quantum} and set coherence \cite{designolle2021set}.

Finally, we apply our results to teleportation fidelities and the theory of quantum discord by defining a global robustness of discord and calculate it analytically  for quantum states close to the set of Bell-diagonal states.

\section{Preliminaries and definitions}\label{sec:pre}

We start by defining the two main types of robustness measures that were previously considered in the literature. In analogy to the terminology used in the context of resource theories, we denote the non-resource carrying states as the free states $\freeset$, and by $\allset$ the set of all quantum states\footnote{Throughout this paper, we assume $\freeset$ and $\allset$ to be (subsets of) quantum states. In the context of resource theories, however, they can be more general, e.g., they can be sets of measurements or channels. Some of our results apply to that case as well.} of dimension $D$. Here, $D$ denotes the total dimension of the state space. In the context of entanglement and discord, it will be useful to also consider the local dimension of the subsystems, which we call $d_A$ and $d_B$, respectively, such that for bipartite systems $D=d_Ad_B$. If the dimensions coincide, we denote the local dimension by $d \equiv d_A = d_B$. For the following definitions of the robustness measures, we are using the nomenclature of \cite{chitambar2019quantum}.

\begin{Def}[Absolute Robustness \cite{RobOrig}] The \emph{absolute robustness} of a quantum state $\rho$ w.r.t.~a set of free states $\freeset$ is given by
\begin{eqnarray}\label{RD}
    R_\freeset(\rho) = \inf_{\sigma \in \freeset} \{ s\geq 0: \rho_s =  \frac{1}{s+1} (\rho + s \cdot \sigma) \in \freeset \}.
\end{eqnarray}
\end{Def}

Note that for the absolute robustness, Eq.~(\ref{RD}), the resilience  of a quantum state $\rho$ is evaluated by mixing it with a free state  $\sigma \in \freeset$. Therefore, if $\freeset$ is of zero measure, the absolute robustness might not be well-defined for many states. One simple example is the case of robustness of purity \cite{horodecki2003reversible, gour2015resource, streltsov2018maximal}, where the set of purity-free states only contains one state, namely the maximally mixed state. In this case, the function as defined in Eq.~(\ref{RD}) would never lead to a finite value for any state $\rho$ -- except of the maximally mixed state itself. That motivates the definition of the global robustness.

\begin{Def}[Global Robustness \cite{GenRob}] The \emph{global robustness} of a quantum state $\rho$ in a resource theory with set of free states $\freeset$ is given by
\begin{eqnarray}\label{genRob}
    R(\rho) = \inf_{\sigma \in \allset} \{ s\geq 0: \rho_s =  \frac{1}{s+1} (\rho + s \cdot \sigma) \in \freeset \}.
\end{eqnarray}
The global robustness is sometimes called generalized robustness.
\end{Def}

Note that there is a third type of robustness, namely the random robustness, where $\sigma$ is restricted to be the maximally mixed state. However, here we focus on the absolute and the global robustness.

Generally speaking, the robustness of a state is defined as the amount of (free) noise that a state tolerates before it could potentially loose all of its resources. These measures have very useful properties. Independently of the specific features of $\freeset$, they are faithful, meaning that
\begin{eqnarray}\label{eq:faithful}
R_{(\freeset)}(\rho) = 0 \Leftrightarrow \rho \in \freeset.
\end{eqnarray}
Another useful property of these measures is their monotonous behavior \cite{RoA}. Given free operations, i.e., a set of completely positive maps, $\{\Gamma_l\}_{l=1}^m$, such that $\sum_l \Gamma_l$ is trace preserving and none of the $\Gamma_l$ can map free states to resourceful states, strong monotonicity holds:
\begin{eqnarray}
    R_{(\freeset)}(\rho) \geq \sum_{l=1}^m \trace[\Gamma_l(\rho)] R_{(\freeset)}\left(\frac{\Gamma_l(\rho)}{\trace[\Gamma_l(\rho)]}\right).
\end{eqnarray}
Specifically, for $m=1$, one obtains the usual monotonicity, i.e., $R_{(\freeset)}(\rho) \geq  R_{(\freeset)}[C(\rho)]$ for all channels $C$ that map free states to free states.

In a quantum resource theory, the free set has typically particular properties, where convexity and star-convexity are probably the most common ones:

\begin{Def}[Convex set]\label{def:convex} A set  $S$ is called \emph{convex}, if any convex combination of two elements of the set $S$  yields another element in the set $S$, i.e.:
\begin{eqnarray}
\forall \rho_1, \rho_2 \in S \,:\, \delta \rho_1 + (1-\delta)\rho_2 \in S \;,  \delta\in [0,1].
\end{eqnarray}
\end{Def}

\begin{Def}[Star-convex set]\label{def:starconvex} A set  $S$ is called \emph{star-convex}, if there exists an element $\rho_0 \in S$, such that any convex combination of it with a state of the set $S$ yields another state in the set $S$, i.e.:
\begin{eqnarray}
    \exists \rho_0\, \forall \rho_1 \in S \,:\,\delta \rho_0 + (1-\delta)\rho_1\in S\;, \delta\in [0,1].
\end{eqnarray}
\end{Def}
For many examples of resource theories with star-convex sets of free states,  $\rho_0$ is given by the maximally mixed state $\rho_0= \frac{\mathds{1}}{D}$.
From Definitions \ref{def:convex} and \ref{def:starconvex}, it is clear that convexity implies star-convexity. 

Our goal is to investigate continuity properties of robustness measures. It is therefore useful to explicitly state the definitions of continuity and the stronger  Lipschitz-continuity for functions of quantum states. Throughout this paper, we will use the trace norm as distance measure between two quantum states via\footnote{Note that some authors define the trace distance of quantum states with an additional factor of $\frac12$.}  $\Vert\rho - \sigma\Vert_{\trace}$. The trace norm is defined as
\begin{eqnarray}
    \Vert X \Vert_{\trace} = \trace(\sqrt{X X^\dagger}), 
\end{eqnarray}
which corresponds to the sum of singular values of $X$. As we restrict ourselves to the finite-dimensional case, whether or not a function is (Lipschitz-) continuous does not depend on the specific choice of distance measure.

\begin{Def}[Continuity]\label{def:cont} A function  $f\,:\,\allset \rightarrow \mathbb{R}$ is \emph{continuous}, if
\begin{eqnarray}
\forall \epsilon>0\, \exists \delta>0:  \Vert \rho_1 - \rho_2\Vert_{\trace} < \delta \implies \vert f(\rho_1) - f(\rho_2) \vert < \epsilon,
\end{eqnarray}
\end{Def}

\begin{Def}[Lipschitz-continuity]\label{def:lipcont} A function  $f\,:\,\allset \rightarrow \mathbb{R}$ is \emph{Lipschitz-continuous} if there exists a constant $L$, called the Lipschitz-constant, such that
\begin{eqnarray}
\forall \rho_1,\rho_2: \vert f(\rho_1)-f(\rho_2)\vert < L \cdot \Vert \rho_1 - \rho_2 \Vert_{\trace},
\end{eqnarray}
\end{Def}
Note that if a function is Lipschitz-continuous with constant $L$, it is also Lipschitz-continuous with $L^\prime > L$. The smallest Lipschitz constant is called optimal. 

\section{Results for general continuity properties}

\begin{table}[t]
    \centering
\caption{Overview of general results of this paper. Here, $\mathcal{B}_\kappa(\sigma_0)$ denotes the closed ball of size $\kappa$ around $\sigma_0$, see Eq.~(\ref{eq:kappaball}), and $\lambda_\text{min}(\sigma_0)$ denotes the smallest eigenvalue of $\sigma_0$. L-continuous stands for Lipschitz-continuous.}
\begin{indented}
    \item[]\resizebox{0.83\textwidth}{!}{
    \begin{tabular}{rr||c|c}
 & & Absolute robustness $R_{\freeset}$  & Global robustness $R$ \tabularnewline
\hline 
\hline 
\multirow{2}{*}{} & \multirow{2}{*}{$\freeset$ convex} & \multicolumn{2}{c}{$\implies$ Continuous on interior of effective~domain}\tabularnewline
 &  & \multicolumn{2}{c}{(Theorem~\ref{convcontthm} \cite{ConvexOpt})}\tabularnewline
\hline 
\multirow{2}{*}{\includegraphics[width=3.8em]{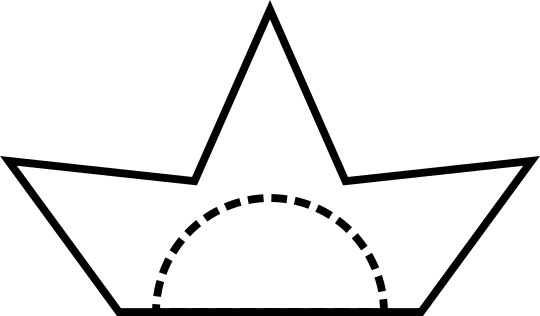}} & $\freeset$ star-convex w.r.t. & \multicolumn{2}{c}{$\implies$ L-continuous, $L=\sfrac{[1-\lambda_{\text{min}}(\sigma_{0})]}{\kappa}$}\tabularnewline
 & all states in $\mathcal{B}_{\kappa}(\sigma_{0})$ & \multicolumn{2}{c}{(Theorem~\ref{thm:kappaball})}\tabularnewline
\hline 
\multirow{2}{*}{\includegraphics[width=2.8em]{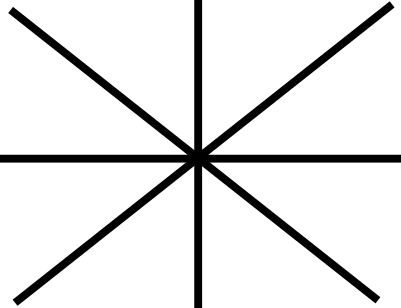}} & $\freeset$ star-convex w.r.t.  & $\nRightarrow$ continuity  & $\implies$ L-continuous, $L=\text{\ensuremath{\sfrac{1}{\lambda_{\text{min}}(\sigma_{0})}}}$\tabularnewline
 & $\sigma_{0}$ of full rank & (Counterexample~\ref{ce:ce1})  & (Theorem~\ref{thm:fullrank})\tabularnewline
\hline 
\multirow{2}{*}{} & \multirow{2}{*}{$\freeset$ general} & $\nRightarrow$ continuity  & $\nRightarrow$  continuity\tabularnewline
 &  & (Counterexample~\ref{ce:ce1})  & (Counterexample~\ref{ce:ce2})\tabularnewline
\end{tabular}
}
\end{indented}
    \label{tab:overview_results}
\end{table}

We are now ready to state our results concerning continuity of the robustness measures defined in Section~\ref{sec:pre}.
An overview of our results can be found in Table~\ref{tab:overview_results}.

We start by reviewing known results for convex free sets from the literature. To that end, we first mention the fact that robustnesses w.r.t~convex free sets $\freeset$ are convex functions themselves \cite{RoA}. In fact, the converse is true as well: If $\freeset$ is non-convex, then the robustnesses are non-convex functions:
    \begin{Thm}
     The absolute robustness  $R_\freeset(\rho)$ and the global robustness $R(\rho)$ are convex if and only if the set $\freeset$ is convex.\end{Thm}
    
\begin{proof}
    As the direction ``$\Leftarrow$'' has been shown in Ref.~\cite{RoA}, we only prove the direction ``$\Rightarrow$''.
    If a set $\freeset$ is not convex, then there exist at least two states $\rho_1,\rho_2 \in \freeset$ and a factor $\alpha \in [0,1]$, such that $\rho_3= \alpha \rho_1 + (1-\alpha) \rho_2$ is not in $\freeset$.
    Then both types of robustness function lead to $R(\rho_1)=R(\rho_2)=0 < R(\rho_3)$, where the inequality follows from faithfulness, Eq.~(\ref{eq:faithful}). This implies that $R(\rho)$ is not convex.
\end{proof}

As stated before, convex functions are known to be continuous under certain (mild) conditions.
A strong version of this fact can be found in Ref.~\cite{ConvexOpt}:

\begin{Thm}[\cite{ConvexOpt}\label{convcontthm}, Proposition 2.17] Every proper convex function $f$ on a finite-dimensional separated
topological linear space $X$ is continuous on the interior of its effective domain.
\end{Thm}

Here, a \emph{proper convex} function is a convex function with non-empty domain that takes values in $\mathbb{R}\,\cup\, \{\infty\}$, but is not equal to $\infty$ on the whole domain. The effective domain is the subset of the domain where the function takes finite values. A separated topological linear space is a vector space defined over a topological space that is Hausdorff, i.e., for each pair of points in it, there exist non-overlapping open neighbourhoods around them. The usual Hilbert space over $\mathbb{C}$ that we work with is an example for such a space. A proof of Theorem~\ref{convcontthm} can be found in \cite{ConvexOpt}.

We now relax the assumptions of convexity of $\freeset$ for the two robustnesses in question. Let us consider the situation that $\freeset$ is star-convex w.r.t.~every state in a non-empty ball \begin{eqnarray}\label{eq:kappaball}
    B_\kappa(\sigma_0) := \{\rho \in \allset\,:\,\Vert\rho - \sigma_0\Vert_{\trace} \leq \kappa\}
\end{eqnarray} around some free state $\sigma_0$.

\begin{Thm}\label{thm:kappaball}
Let $\freeset$ be a subset of $\allset$ and $\sigma_0 \in \freeset$, such that there exists $\kappa > 0$ with $B_\kappa(\sigma_0) \subset \freeset$. If $\freeset$ is star-convex w.r.t.~each state in $B_\kappa(\sigma_0)$, then
\begin{enumerate}
    \item[(a)]  $R(\rho) \leq R_\freeset(\rho) \leq \frac{2(1-\lambda_\text{min}(\sigma_0))}{\kappa} - 1$,  where $\lambda_\text{min}(\sigma_0)$ denotes the smallest eigenvalue of $\sigma_0$;
    \item[(b)] $R$ and $R_\freeset$ are Lipschitz-continuous on $\allset$ with Lipschitz constant $L=\frac{1-\lambda_\text{min}(\sigma_0)}{\kappa}$.
\end{enumerate}
\end{Thm}
\noindent Note that $\sigma_0$ might lie on the boundary of $\allset$. The proof of Theorem~\ref{thm:kappaball} is given in \ref{sec:proofthm3}.

An example is given by the robustness of entanglement, where it is known from Ref.~\cite{Kappaball} that a bipartite state $\rho$ is certainly separable, if  $\Vert \rho - \frac{\mathds{1}}{d_Ad_B} \Vert_{\trace} \leq \sqrt{\frac{1}{d_Ad_B(d_Ad_B-1)}}$. Thus, we can choose $\kappa = \frac1{\sqrt{d_Ad_B(d_Ad_B-1)}}$ and $\sigma_0 = \frac{\one}{d_Ad_B}$, and from Theorem~\ref{thm:kappaball} we get
\begin{eqnarray}
   L = \sqrt{\frac{(d_Ad_B-1)^3}{d_Ad_B}} \in \mathcal{O}(d_Ad_B).
\end{eqnarray}

Of course, the set of separable states fulfills more properties than the assumptions from Theorem~\ref{thm:kappaball} require, such that we can find a smaller Lipschitz constant for the absolute robustness. In fact, a much tighter dependence of $L$ on the local dimensions can be found.
\begin{Thm}\label{thm:roe}
Let $\freeset$ be the set of separable states in a bipartite system $\allset$. Then $R_\freeset$ is Lipschitz-continuous on $\allset$ with Lipschitz constant $L = \min(d_A,d_B)-\frac12$.
\end{Thm}
The proof is located in \ref{sec:proofthm4}.

The question arises, whether the requirement of the existence of a $\kappa$-ball around $\sigma_0$ in Theorem~\ref{thm:kappaball} can be removed. For the absolute robustness, the answer is negative, as the following counterexample shows.

\begin{figure}
    \centering
    \includegraphics[width=0.8\columnwidth]{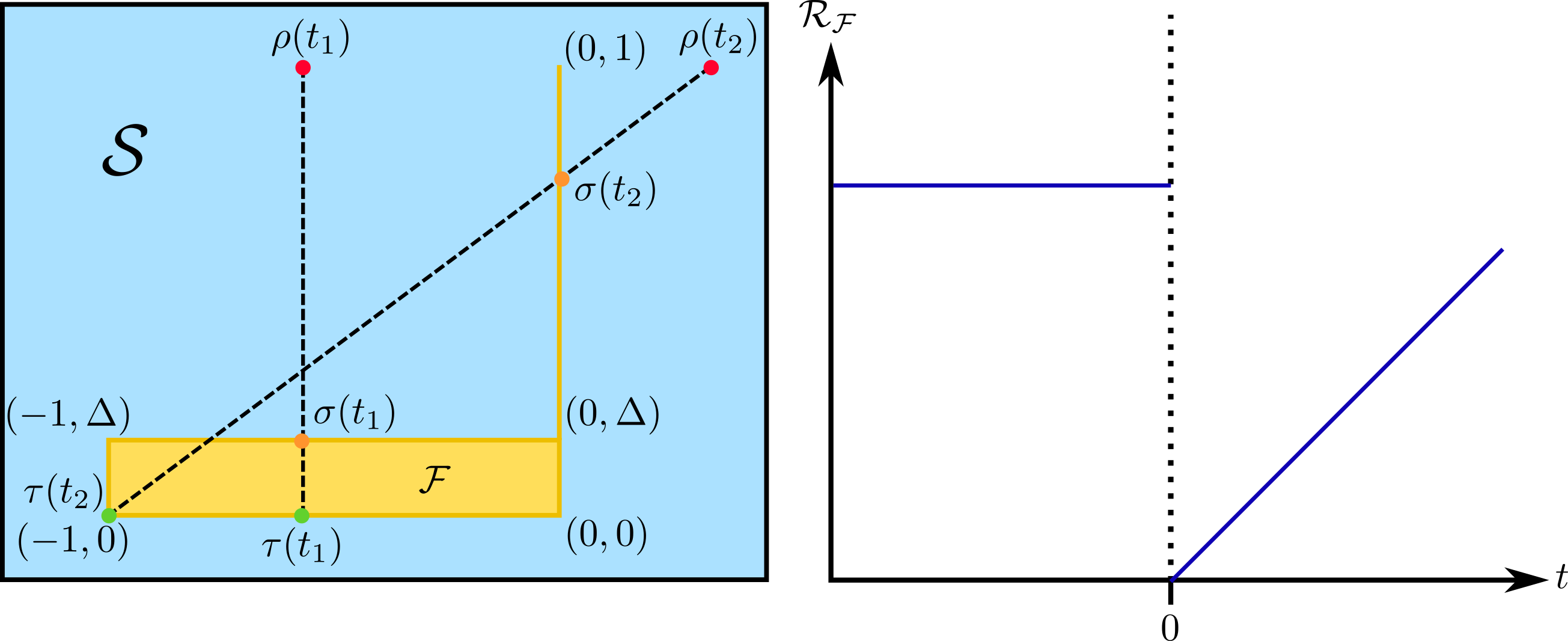}
    \caption{(Counterexample~\ref{ce:ce1}) A two-dimensional projection of a star-convex free set (orange) that leads to an absolute robustness that is non-continuous at the point $(0,1)$. The optimal free states for the choices of $\rho(t_1) = (t_1,1)$ to the left and $\rho(t_2) = (t_2,1)$ to the right of that point are displayed as well. See the text in Counterexample~\ref{ce:ce1} for more details.}
    \label{fig:counterexample1}
\end{figure}

\begin{counterexample}[Non-continuity of absolute robustness]\label{ce:ce1}
\emph{Consider the two-dimensional representation of a star-convex set of free states depicted in Fig.~\ref{fig:counterexample1}, such that the states are parameterized by two coordinates $(x,y)$. $\freeset$ consists of a line from $(0,0)$ to $(0,1)$ and a perpendicular strip of width $\Delta$, bounded by points $(0,0)$, $(-1,0)$, $(-1,\Delta)$ and $(0,\Delta)$. We consider the family of ``states'' $\rho(t) = (t,1)$ where $t\in \mathbb{R}$. In order to find the robustness of theses states, we have to find all points $\tau(t), \sigma(t) \in \freeset$, s.t. $\rho(t) + s\tau(t) = (1+s)\sigma(t)$.}

\emph{We start with the case $t_1<0$. Note that the robustness $s$ is the smallest attainable ratio between the distance between $\sigma$ and $\tau$, and the distance between $\rho$ and $\sigma$. Thus, in order to minimize $s$, it is best to choose $\sigma(t_1)$ on the top boundary of the stripe, and $\tau(t_1)$ on the bottom, such that a straight line connects $\rho(t_1)$, $\tau(t_1)$ and $\sigma(t_1)$.}

\emph{Each such choice for $\sigma(t_1)$ and $\tau(t_1)$ yields the same value for $s$, as they can be transformed into each other by shearing, which leaves the ratio between the distance between $\rho(t_1)$ and $\sigma(t_1)$, and the distance between $\sigma(t_1)$ and $\tau(t_1)$ the same, and therefore also $s$. We therefore choose $\sigma(t_1) = (t_1,\Delta)$ and $\tau(t_1) = (t_1, 0)$ and obtain $R_{\freeset}[\rho(t_1)] = \frac{1-\Delta}{\Delta}$ for $t_1 < 0$, independent of $t_1$.}

\emph{For $t_2\geq 0$, the optimal choice for $\sigma(t_2)$ now lies on the vertical line instead, whereas $\tau(t_2)$ sits on the left-hand side boundary of the strip. Again, due to shearing, all admissible pairs of $\tau(t_2)$ and $\sigma(t_2)$ yield the same $s$. We choose $\tau(t_2) = (-1,0)$ and $\sigma(t_2) = (0, \frac{1}{1+t_2})$, which yields $R_{\freeset}[\rho(t_2)] = t_2$ for $t_2 \geq 0$. Thus, the robustness is non-continuous at the point $(0,1)$.}
\end{counterexample}

Let us now turn to the case of the global robustness. As it turns out, the existence of a $\kappa$-ball in $\freeset$ is not required in this case.

\begin{Thm}\label{thm:fullrank}
Let $\freeset$ be a subset of $\allset$ that is star-convex w.r.t.~a fixed quantum state $\sigma_0$ of full rank, i.e., with smallest eigenvalue $\lambda>0$. Then the global robustness $R$ is Lipschitz-continuous on its effective domain with Lipschitz constant $L=\frac1\lambda$.
\end{Thm}
The proof can be found in \ref{sec:proofthm5}.

\begin{figure}
    \centering
    \includegraphics[width=0.98\columnwidth]{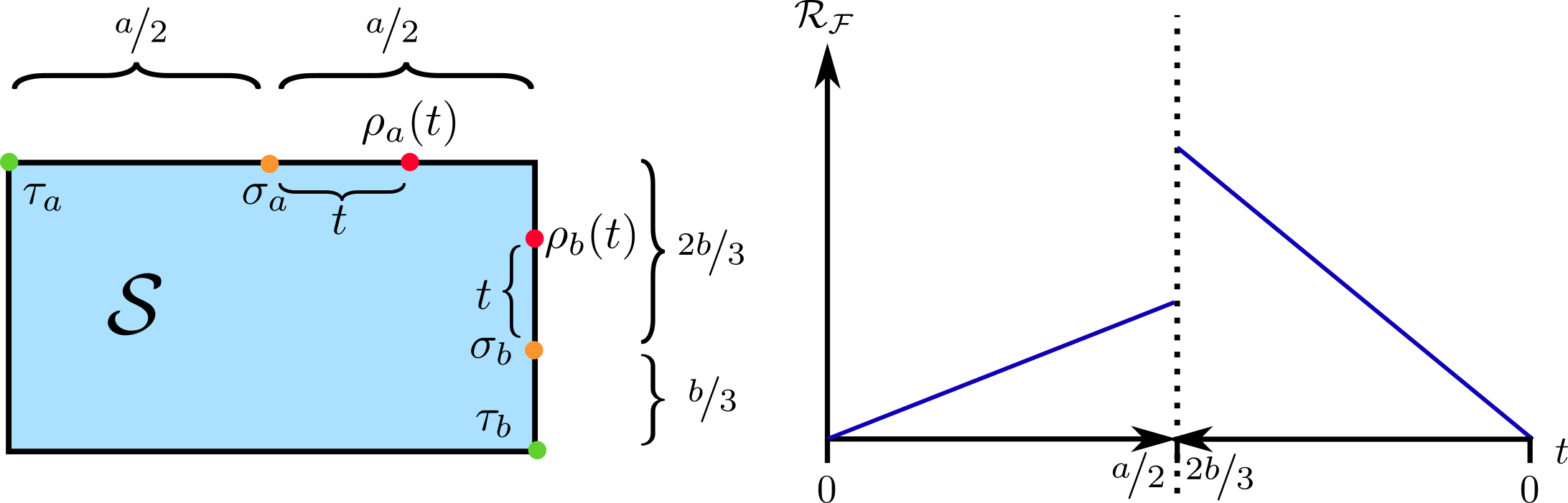}
    \caption{(Counterexample~\ref{ce:ce2}) A free set consisting of two points $\sigma_a$ and $\sigma_b$ at the boundary of the convex set $\allset$ that leads to a non-continuous global robustness on the boundary of $\allset$. On the interior, the robustness is infinite. For more details, see the text in Counterexample~\ref{ce:ce2}. }
    \label{fig:counterexample2}
\end{figure}

Several resource theories are based on sets that are star-convex w.r.t.~the maximally mixed state $\sigma_0=\frac{\mathds{1}}{D}$. This state has its smallest eigenvalue equal to $\lambda=\frac{1}{D}$. According to Theorem~\ref{thm:fullrank}, this leads to a Lipschitz constant of $L = D$ for the global robustness $R$. However, Theorem~\ref{thm:fullrank} is very general, and additional properties of $\freeset$, such as convexity, can lead to robustness measures with much lower optimal Lipschitz constants.

An immediate question that arises is whether the global robustness defined for a convex state space $\allset$ is continuous, regardless of the specific choice of the free set $\freeset$. The answer is negative, as the following counter-example demonstrates.

\begin{counterexample}[Non-continuity of global robustness]\label{ce:ce2}\theoremstyle{definition}
\emph{Consider the convex set as in Fig.~\ref{fig:counterexample2}, where the free set consists of two states, $\sigma_a$, lying on the center of a flat boundary of length $a$, and $\sigma_b$, lying off-center of a second flat boundary of length $b$, such that the two boundaries touch at some angle.}

\emph{Now, in the interior of $\allset$, the global robustness is clearly infinite, as no proper mixture of any two states can reach $\sigma_a$ or $\sigma_b$. On the two straight boundaries, however, such a decomposition is possible. The indicated family of states, $\rho_a(t)$, can be mixed with $\tau_a$ to reach $\sigma_a$, yielding $R[\rho_a(t)] = \frac{2t}{a}$, whereas family $\rho_b(t)$ has robustness $R[\rho_b(t)] = \frac{3t}{b}$. However, $\rho_a(\frac{a}{2}) = \rho_b(\frac{2b}{3})$, but $1=R[\rho_a(\frac{a}{2})] \neq R[\rho_b(\frac{2b}{3})]=2$.}
\end{counterexample}

\section{Examples of specific resource theories}

To conclude, we list as examples  resource theories together with properties of their free sets in Table~\ref{tab:2}. On top of the list, we have well-studied resource theories such as entanglement, purity, and coherence. The resource theories of set coherence, imaginarity and discord have also been previously introduced in the literature. Moreover, resource theories such as the resource of negativity can naturally be defined from the free set of states with positive partial transposition.

Recently, it was conjectured~\cite{yamasaki2022activation} and subsequently shown \cite{palazuelos2022genuine} that multiple copies of any multipartite state that is not separable w.r.t.~a fixed partition leads to genuine multipartite entanglement (GME). Thus, a resource theory where multi-copy GME is a resource would have the set of states separable w.r.t.~a fixed partition as a maximal free set, as listed in the table. This latter set was also recently identified as a subset of free states for conference key agreement protocols, i.e., states from which no conference key can be extracted \cite{carrara2021genuine}. Finally, we list the resource theory of teleportability, which we introduce in the next section. The global robustness is continuous for all the listed resources, since every one of them is based on a free set that is at least star-convex w.r.t.~the maximally mixed state.

In the following, we investigate two of the resources in more detail. First, we will use our results to introduce a measure of robustness of teleportability. Second, we construct a continuous global robustness w.r.t.~the star-convex free set of states of vanishing discord.

\subsection{Robustness of teleportability} \label{subsec:teleportability}

\begin{table}[t]
    \centering
     \caption{Examples of resource  theories and properties of their free sets.}
    \begin{indented}
    \item[]\resizebox{0.83\textwidth}{!}{
\begin{tabular}{l|l|l}
    Robustness of ... & set $\freeset$  & convexity? \\
    \hline
     entanglement \cite{RobOrig, GenRob}    & separable states         & convex    \\
     purity \cite{horodecki2003reversible, gour2015resource, streltsov2018maximal}          & the maximally mixed state   & convex    \\
     coherence w.r.t.~fixed basis \cite{napoli2016robustness, RoA}       &incoherent states        & convex    \\
     discord (Sec.~\ref{subsec:rod})    & states without discord   & star-convex    \\
     negativity & states with positive partial transpose& convex    \\
     set coherence \cite{designolle2021set}  & tuples of states that commute & star-convex\\
     fixed partition non-biseparability & states that are bisep.~w.r.t.~fixed bip. & star-convex\\
     imaginarity \cite{Hickey18}&real quantum states&convex\\
     teleportability (Sec.~\ref{subsec:teleportability})& unfaithful states & convex
    
\end{tabular}
}
\end{indented}
   
    \label{tab:2}
\end{table}
In a teleportation scheme, two parties use a pre-shared entangled state to send a single-party quantum state from one party to another by performing measurements on the state to be sent and one part of the entangled state, and communicating the measurement results classically. The maximal fidelity of the received state is determined by the quantity \cite{horodecki1999general, guhne2021geometry}
\begin{eqnarray}
    f_\text{max} = \frac{dF_\text{max} + 1}{d+1},
\end{eqnarray}
where $d$ denotes the local dimension of the entangled pair and
\begin{eqnarray}
    F_\text{max}(\rho) = \max_{U,V\in \text{SU}(d)} \braket{\phi^+|U^\dagger \otimes V^\dagger\rho U\otimes V |\phi^+}
\end{eqnarray} is the singlet fraction of $\rho$ with $\ket{\phi^+} = \frac1{\sqrt{d}}\sum_{i=0}^{d-1} \ket{ii}$. Here, $\rho$ is the $d^2$-dimensional entangled shared state, and the maximization is performed over all local unitary transformations. 

The best fidelity using classical protocols is given by $\frac2{d+1}$, implying that a specific shared state is useful for a teleportation task, whenever $F_\text{max}(\rho)> \frac1d$. This immediately suggests a definition of a set of $\freeset_\text{tel} = \{\rho \,:\, F_\text{max}(\rho) \leq \frac1d\}$ of useless states. We can thus define the robustness of teleportability as
\begin{eqnarray}\label{eq:rtel}
R_\text{tel}(\rho) = \inf_{\sigma \in \allset} \{ s\geq 0: \rho_s =  \frac{1}{s+1} (\rho + s \cdot     \sigma) \in \freeset_\text{tel} \}.
\end{eqnarray}

Using Theorem~\ref{thm:kappaball}, we get the following.
\begin{Crl}\label{crl:fidelity}
    Let $R_\text{tel}$ denote the global robustness of teleportability defined in Eq.~(\ref{eq:rtel}). Then
    \begin{eqnarray}
        \vert R_\text{tel}(\rho) - R_\text{tel}(\tau) \vert \leq (d+1)\Vert \rho - \tau \Vert_{\trace}
    \end{eqnarray}
    for all quantum states $\rho$ and $\tau$.
\end{Crl}
The proof is located in \ref{app:prooffidelity}.
Note that Theorem~\ref{thm:kappaball} also allows us to extract the upper bound of $R_\text{tel}(\rho) \leq 2d+1$ for all $\rho$.

\subsection{Robustness of discord} \label{subsec:rod}

Quantum discord is a non-symmetric measure of non-classicality of correlations in quantum systems, which coincides with entanglement for pure states, but has different features for mixed states. In fact, the quantum discord of a state vanishes exactly if there exists a measurement on one party that does not disturb the system. For a detailed description of discord, see Ref.~\cite{bera2017quantum} and references therein.

In its original definition, the discord of a bipartite quantum state $\rho_{AB}$ is defined as
\begin{eqnarray}
    \mathcal{D}_B(\rho_{AB}) = S(\rho_B) - S(\rho_{AB}) + \min_{\{\Pi_k^B\}} \sum_k p_k S(\rho_{A|k}),
\end{eqnarray}
where the minimization is performed over all projective measurements on Bob's system, and $p_k = \trace(\Pi_k^B \rho_{AB})$, $\rho_{A|k} = \Pi_k^B \rho_{AB} \Pi_k^B / p_k$.

Due to the minimization, evaluating the discord for a general quantum state is hard and was so far only accomplished for certain two-qubit states. Thus, other measures for quantum discord were proposed. However, the non-convexity of the set of zero discord states complicates the task of defining such measures. Here, we define the global robustness of discord. We start by pointing out that the set of zero discord states is star-convex w.r.t.~the maximally mixed state. To that end, note that it has been shown that states with vanishing discord are the so-called quantum-classical states, given by
\begin{eqnarray}
    \sigma_{AB} = \sum_{i=1}^{r} p_i \rho_i \otimes \ketbra{\psi_i}{\psi_i},
\end{eqnarray}
where the $\ket{\psi_i}$ are orthogonal \cite{bera2017quantum}. In case $r < d$, we extend the set of states to a complete basis of Bob's Hilbert space, such that $\one_B = \sum_{i=1}^{d} \ketbra{\psi_i}{\psi_i}$. A mixture of $\sigma_{AB}$ with the maximally mixed state yields then
\begin{eqnarray}
    (1-p)\sigma_{AB} + \frac{p}{d^2}\one_A\otimes \one_B = \sum_{i=1}^r [(1-p)p_i\rho_i& + \frac{p}{d^2}\one_A] \otimes \ketbra{\psi_i}{\psi_i}+ \nonumber \\
    & + \sum_{i=r+1}^d \frac{p}{d^2} \one_A \otimes \ketbra{\psi_i}{\psi_i},
\end{eqnarray}
which is again a quantum-classical state.
While the set $\freeset$ is star-convex, it lacks a $\kappa$-ball around the maximally mixed state. This can be seen by exhibiting the discord of two-qubit Werner states, $\rho_{AB} = (1-p)\ketbra{\psi^-}{\psi^-} + \frac{p}{d^2}\one_{AB}$, which has nonzero discord whenever $p>0$ \cite{ollivier2001quantum}.
Thus, we define a global robustness measure for discord,

\begin{figure}
    \centering
    \includegraphics[width=0.6\columnwidth]{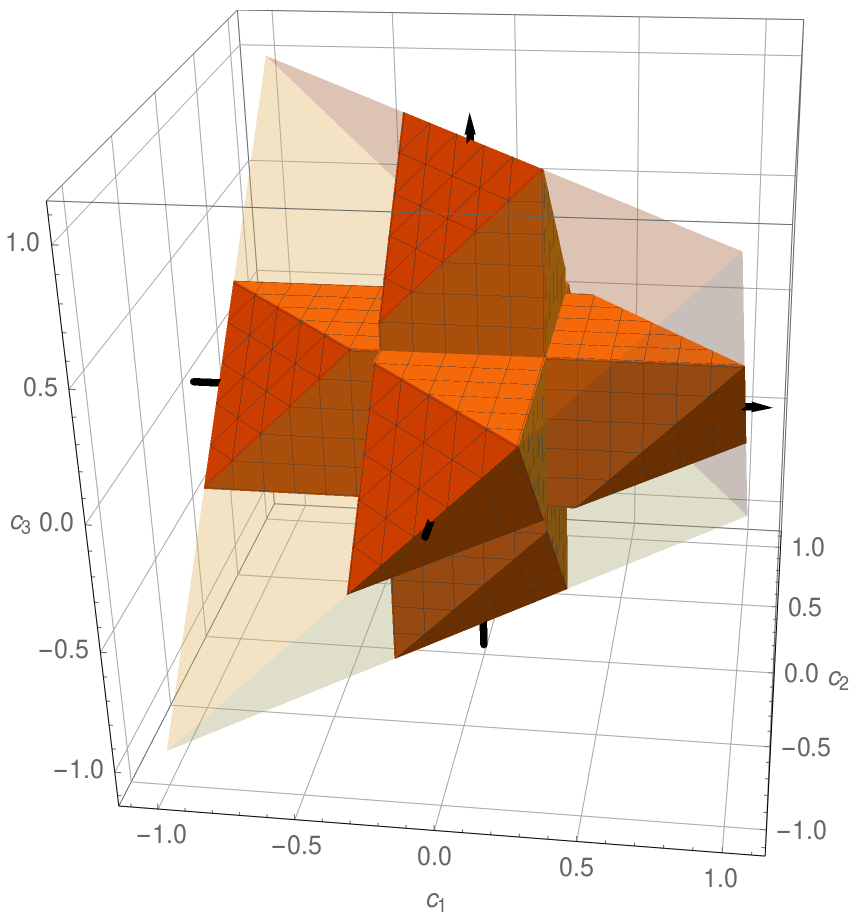}
    
    \caption{The tetrahedron of Bell-diagonal states (transparent yellow) and the subset of Bell-diagonal states with robustness of discord of less than or equal to $0.3$ (solid orange). The states of zero discord are located on the coordinate axes, displayed here in solid black.}
    \label{fig:rod}
\end{figure}

\begin{eqnarray}
    R_{\mathcal{D}_B}(\rho) := \inf_{\sigma\in \allset}\{s\geq 0 : \frac{1}{1+s}(\rho + s\sigma) \in \freeset_{\mathcal{D}_B}\},
\end{eqnarray}
where $\freeset_{\mathcal{D}_B}$ denotes the set of zero discord states. According to Theorem~\ref{thm:fullrank}, this function is Lipschitz-continuous with Lipschitz constant bounded by $L=d^2$. Note that evaluation of $R_{\mathcal{D}_B}$ is still complicated, as it is hard to decide whether a given state is quantum-classical or not. However, one can replace this condition by an SDP hierarchy, giving a (converging) outer approximation of the set $\freeset_{\mathcal{D}_B}$ \cite{piani2016hierarchy}. With this hierarchy, we could efficiently find lower bounds on $R_{\mathcal{D}_B}(\rho)$. In the case of Bell-diagonal two-qubit states, however, we can analytically calculate the robustness. These states are given by the three-parameter family
\begin{eqnarray}
    \rho_{\text{BDS}} = \frac14(\one\otimes\one + \sum_{i=1}^3 c_i \sigma_i \otimes \sigma_i),
\end{eqnarray}
where $\sigma_1, \sigma_2, \sigma_3$ are the usual Pauli matrices, and positivity of $\rho_{\text{BDS}}$ demands
\begin{eqnarray}
    1-c_1-c_2-c_3 \geq 0,&\quad 1-c_1+c_2+c_3 \geq 0,\\
    1+c_1-c_2+c_3 \geq 0,&\quad 1+c_1+c_2-c_3 \geq 0,
\end{eqnarray}
giving rise to the tetrahedron displayed transparently in Fig.~\ref{fig:rod}. For this set, the discord can be calculated analytically, and by using Lagrange multipliers it follows immediately that it vanishes iff two of the $c_i$ vanish \cite{luo2008quantum}. In other words, the free set consists of the union of the three coordinate axes. In order to perform the optimization, we will use the equivalent formulation \cite{RoA}
\begin{eqnarray}\label{eq:robdisalt}
    R_{\mathcal{D}_B}(\rho) = \inf_{\sigma \in \freeset} \{s\geq 0 : \rho \leq (1+s)\sigma\}.
\end{eqnarray}
However, a priori it is not clear whether we can restrict the optimization to the set of Bell-diagonal free states. We show this in the following Lemma, which is proven in \ref{app:belldiagonal}.

\begin{Lem}\label{lem:belldiagonal}
If $\rho$ is a Bell diagonal state, then $R_{\mathcal{D}_B}(\rho)$ can be computed, via Eq.~(\ref{eq:robdisalt}),  by restricting the infimum to Bell diagonal  free states.
\end{Lem}

This allows us to perform the optimization in Eq.~(\ref{eq:robdisalt}) and, by using Theorem~\ref{thm:fullrank}, we can bound the robustness of discord for other states as well:
\begin{Crl}\label{crl:rod}
    Let $\rho = \frac14(\one \otimes \one + \sum_{i=1}^3 x_i \sigma_i \otimes \one + \sum_{i=1}^3 y_j \one \otimes \sigma_j + \sum_{i,j=1}^3 T_{ij}\sigma_i\otimes\sigma_j)$. Then the robustness of discord $R_{\mathcal{D}_B}(\rho)$ is bounded by
    \begin{eqnarray}
        \vert c_2 \vert - 4\max(\vert x\vert, \vert y\vert) \leq R_{\mathcal{D}_B}(\rho) \leq \vert c_2\vert + 4\max(\vert x\vert, \vert y\vert),
    \end{eqnarray}
    where $\vert c_2\vert$ is the singular value of the correlation matrix $T$ with the second largest magnitude and $\vert x\vert$ and $\vert y \vert$ denote the vector norms of the local Bloch vectors of $\rho_A$ and $\rho_B$, respectively.
    
    In particular, if $\rho_\text{BDS}$ is Bell-diagonal, then
    \begin{eqnarray}
        R_{\mathcal{D}_B}(\rho_\text{BDS}) = \vert c_2 \vert.
    \end{eqnarray}
\end{Crl}
For the proof, refer to \ref{app:proofrod}.

\section{Conclusion}

We investigated continuity properties of the absolute and global robustness measures for different types of free sets of quantum states. While it was known before that for convex free sets, the robustness functions are continuous, we gained new insights in two directions: first, we reduced the assumptions on the free set, allowing for star-convexity. Second, we proved a stronger form of continuity, namely Lipschitz-continuity, which allows to derive bounds on the difference of robustness measures of quantum states based only on their distance. This may be of practical relevance in experiments, where a limit on the amount of errors could be translated into a maximal deviation of the robustness.

We also showed that it is unlikely that the remaining restrictions on the free set can be reduced further without loosing the property of continuity, by providing explicit examples of sets which lead to non-continuous robustnesses.

Finally, we investigated two specific examples, namely the convex robustness of teleportability and the star-convex discord. For the former, we found a $\kappa$-ball to calculate the Lipschitz constant of the absolute and global robustness of teleportability. For the latter, we defined the robustness of discord, evaluated it explicitly for the case of Bell-diagonal states and gave bounds for general two-qubit states.

Future directions of research include the study of other robustness measures, like the random robustness, whose continuity features could be investigated. Additionally, a complete characterization of the features of the free set leading to continuous robustness measures would yield deeper insight in similarly defined measures.

\ack
We thank Thomas Wagner and Chau Nguyen for fruitful discussions. 
We acknowledge financial support by the QuantERA project QuICHE via the German Ministry of Education and Research (BMBF
Grant No.~16KIS1119K) and by the
Deutsche Forschungsgemeinschaft (DFG, German Research Foundation) under Germany’s Excellence Strategy - Cluster of Excellence Matter and Light for Quantum Computing (ML4Q) EXC 2004/1 - 390534769.
\section*{References}
 
\bibliographystyle{unsrt}
\bibliography{biblio}
\newpage

\appendix

\section{Proof of Theorem \ref{thm:kappaball}}\label{sec:proofthm3}

\setcounter{Thm}{2}
\begin{Thm}
Let $\freeset$ be a subset of $\allset$ and $\sigma_0 \in \freeset$, such that there exists $\kappa > 0$ with $B_\kappa(\sigma_0) \subset \freeset$. If $\freeset$ is star-convex w.r.t.~each state in $B_\kappa(\sigma_0)$, then
\begin{enumerate}
    \item[(a)]  $R(\rho) \leq R_\freeset(\rho) \leq \frac{2(1-\lambda_\text{min}(\sigma_0))}{\kappa} - 1$,  where $\lambda_\text{min}(\sigma_0)$ denotes the smallest eigenvalue of $\sigma_0$;
    \item[(b)] $R$ and $R_\freeset$ are Lipschitz-continuous on $\allset$ with Lipschitz constant $L=\frac{1-\lambda_\text{min}(\sigma_0)}{\kappa}$.
\end{enumerate}
\end{Thm}

\begin{proof}
As $B_\kappa(\sigma_0) \subset \freeset$, we can mix every state $\rho\in \allset$ with $\sigma_0$ until it becomes member of the free set. In particular, for fixed $\rho$, consider the family 
\begin{equation}\label{eq:mixedsigma0}
\rho(p) = p\rho + (1-p)\sigma_0.
\end{equation}
It certainly becomes a member of $\freeset$, whenever
\begin{eqnarray}
    \kappa \geq \Vert \rho(p) - \sigma_0 \Vert_{\trace} = p \Vert \rho - \sigma_0 \Vert_{\trace}.
\end{eqnarray}
We can maximize the r.h.s.~over all quantum states $\rho$. Note that the trace norm is invariant under unitary transformations, thus we can assume that $\sigma_0 = \diag(\lambda_1, \lambda_2,\ldots,\lambda_D)$, where $\lambda_1 \leq \lambda_2 \leq \ldots \leq \lambda_D$. Optimization over all states $\rho$ readily yields a maximizing state $\rho$ which corresponds to the eigenvector of $\sigma_0$ w.r.t.~its minimal eigenvalue $\lambda_1$. Thus we can find an upper bound for the value of $p$, always leading to a free state $\rho(p)$:
\begin{eqnarray}
    p \leq \frac{\kappa}{\max_\rho \Vert \rho - \sigma_0 \Vert_{\trace}} = \frac{\kappa}{(1-\lambda_1) + \sum_{i=2}^d \lambda_i} = \frac{\kappa}{2(1-\lambda_1)}.
\end{eqnarray}
We set $p_0 = \frac{\kappa}{2(1-\lambda_1)}$, such that for each quantum state $\rho$, $\rho(p \leq p_0) \in \freeset$. This automatically yields an upper bound on $R$ and $R_\freeset$: Let $\rho$ be some quantum state, then $p_0\rho + (1-p_0)\sigma_0$ is member of $\freeset$ and choosing $s$ such that $\frac{1}{1+s} = p_0$ yields a valid point in the set that is optimized over to obtain $R$ and $R_\freeset(\rho)$. Therefore, $R(\rho) \leq R_\freeset(\rho) \leq \frac{1}{p_0}-1 = \frac{2(1-\lambda_1)}{\kappa} - 1$, establishing (a).

Now we are in position to prove Lipschitz-continuity. Let $\rho$ and $\tau$ be two quantum states with $\Vert\rho - \tau\Vert_{\trace} = \delta$. This implies that we can find two quantum states $\rho_a, \rho_b$, such that we can write (see also \ref{sec:proofthm4})
\begin{eqnarray}\label{eq:tauminusrho}
    \tau - \rho = \frac\delta2(\rho_a - \rho_b) = \frac{\delta}{2p_0}(\rho_a(p_0) - \rho_b(p_0)),
\end{eqnarray}
with $\rho_{a,b}(p_0)$ being defined in Eq.~(\ref{eq:mixedsigma0}). Note that $\rho_{a,b}(p_0) \in \freeset$.

Let $R_\freeset(\rho) = s^\star$ ($R(\rho) = s^\star$), and its optimal decomposition be given by
\begin{eqnarray}
    \frac{\rho + s^\star \tilde{\sigma}^\star}{1+s^\star} = \sigma^\star 
\end{eqnarray}
for some states $\tilde{\sigma}^\star \in \freeset$ ($\tilde{\sigma}^\star \in \allset$) and  $\sigma^\star \in \freeset$.\footnote{Here, some care has to to be taken if $\freeset$ is not closed. In that case, optimal choices for $\tilde{\sigma}^\star$ and $\sigma^\star$ might not exist. However, in that case, we can choose them in the closure $\bar{\freeset}$ of $\freeset$, and the rest of the argument can be adapted accordingly, as the robustnesses are defined via the infimum.} Inserting this optimal decomposition for $\rho$ into Eq.~(\ref{eq:tauminusrho}), yields
\begin{eqnarray}
    \tau + (s^\star + \frac{\delta}{2p_0})\frac{s^\star \tilde{\sigma}^\star + \frac{\delta}{2p_0} \rho_b(p_0)}{s^\star + \frac{\delta}{2p_0}} = (1+s^\star + \frac{\delta}{2p_0})\frac{(1+s^\star)\sigma^\star + \frac{\delta}{2p_0}\rho_a(p_0)}{1+s^\star + \frac{\delta}{2p_0}}.\nonumber\\
    ~
\end{eqnarray}
As $\freeset$ is star-convex w.r.t all states in $B_\kappa(\sigma_0)$, also the convex combinations of $\tilde{\sigma}^\star$ and $\rho_b$, as well as $\sigma^\star$ and $\rho_a$ are members of $\freeset$ (in case of $R$, this is true only for the mixture of $\sigma^\star$ and $\rho_a$). Therefore, the above equation yields a proper decomposition of $\tau$ in the set that is optimized over to obtain $R_\freeset(\tau)$ ($R(\tau)$). Therefore, $R_\freeset(\tau) \leq s^\star + \frac{\delta}{2p_0} = R_\freeset(\rho) + \frac{\delta}{2p_0}$ ($R(\tau) \leq  R(\rho) + \frac{\delta}{2p_0}$). Mutatis mutandis, we obtain the same with switched roles of $\rho$ and $\tau$, yielding
\begin{equation}
    \vert R_\freeset(\rho) - R_\freeset(\tau) \vert \leq \frac{\delta}{2p_0} = \frac{1-\lambda_1}{\kappa} \Vert \rho - \tau \Vert_{\trace}
\end{equation}
(and the same for $R$), such that Lipschitz-constant is given by
\begin{eqnarray}
    L = \frac{1-\lambda_1}{\kappa},
\end{eqnarray}
where $\lambda_1$ denotes the smallest eigenvalue of $\sigma_0$, and thus showing (b).
\end{proof}

\section{Proof of Theorem \ref{thm:roe}}\label{sec:proofthm4}

\setcounter{Thm}{3}
\begin{Thm}
Let $\freeset$ be the set of separable states in a bipartite system $\allset$. Then $R_\freeset$ is Lipschitz-continuous on $\allset$ with Lipschitz constant $L = \min(d_A,d_B)-\frac12$.
\end{Thm}
    \begin{proof}
    We are considering the robustness of  entanglement. In this case, the maximal value the function $R_\freeset$ can reach is bounded by $\min(d_A,d_B)-1$ \cite{RobOrig}.
    Consider two states $\rho_1$ and $\rho_2$, which have $\Vert \rho_1-\rho_2\Vert_{\trace}=\delta$. Then, one can write $\rho_2=\rho_1 + (P - N)$, with $P, N$ being positive, orthogonal, hermitian operators with $\trace(P)=\trace(N)=\frac{\delta}2$, as $\Vert \rho_1 - \rho_2\Vert_{\trace}=\trace (P)+\trace(N)=\delta$.
    Being positive and hermitian, we can write $P= \frac\delta2 \rho_a$ and $N= \frac\delta2 \rho_b$ as multiples of quantum states $\rho_a$ and $\rho_b$, and therefore
    \begin{eqnarray}\label{ro2}
        \rho_2=\rho_1 +\frac{\delta}2 (\rho_a-\rho_b).
    \end{eqnarray}
     
     We know that the absolute robustness of entanglement is bounded and finite everywhere \cite{RobOrig}, and the set of separable states is closed. Thus, for every quantum state, we can find its optimal decomposition such that we can write
     \begin{eqnarray}
         \rho_{i} + s_i\tilde{\sigma}_1 = (1+s_i)\sigma_i,
     \end{eqnarray}
     where $s_i = R_\freeset(\rho_i)$, $\sigma_i, \tilde{\sigma}_i \in \freeset$ and $i\in\{1,2,a,b\}$.
     Solving this for the states $\rho_1$, $\rho_2$, $\rho_a$ and $\rho_b$, leads to the pseudo-mixtures
    \begin{eqnarray}\label{nonconvex}
        \rho_i = (1+s_i) \sigma_i - s_i \tilde{\sigma_i},
    \end{eqnarray}
    and inserting the corresponding expressions for $\rho_1$, $\rho_a$ and $\rho_b$ in Eq.~(\ref{ro2}) leads to
    \begin{eqnarray} \label{upperRo2}
        \rho_2 = [(1+s_1) \sigma_1 + \frac\delta2 ((1+s_a) \sigma_a + s_b \tilde{\sigma}_b)] - [s_1 \tilde{\sigma}_1 + \frac\delta2 (s_a \tilde{\sigma}_a + (1+s_b) \sigma_b].\nonumber\\
        ~
    \end{eqnarray}
     
     Now, one can compare the expressions for $\rho_2$ in Eqs.~(\ref{nonconvex}) and (\ref{upperRo2}) to write down a pseudomixture of $\rho_2$ that will lead to an upper bound on $R_\freeset(\rho_2)=s_2$. We set 
     \begin{eqnarray}
        \sigma^\prime_2 &:= \frac{(1+s_1) \sigma_1 + \frac\delta2 ((1+s_a) \sigma_a + s_b \tilde{\sigma}_b)}{1+s_1+\frac\delta2(1+s_a+s_b)}, \\
        \tilde{\sigma}^\prime_2 &:= \frac{s_1 \tilde{\sigma}_1 + \frac\delta2 (s_a \tilde{\sigma}_a + (1+s_b) \sigma_b}{s_1 + \frac\delta2 ( 1 + s_a + s_b )},
    \end{eqnarray}
    and $s^\prime_2= s_1 + \frac\delta2 (1 + s_a + s_b)$, which is an upper bound to $s_2$, yielding
    \begin{eqnarray} \label{bound}
        s_2-s_1 \leq \frac\delta2 (1 + s_a + s_b).
    \end{eqnarray}
    
    Exchanging the roles of $\rho_1$ and $\rho_2$ (and therefore, $\rho_a$ and $\rho_b$, we obtain that Eq.~(\ref{bound}) is also true for the absolute value of $s_2-s_1$.
    
    As noted before, the values $s_a$ and $s_b$ both are bounded by $\min(d_A,d_B)-1$. Therefore, 
    \begin{eqnarray} \label{firstcomp}
        \vert s_2 - s_1 \vert \leq \frac\delta2 ( 2 \min(d_A,d_B) - 1).
    \end{eqnarray}
    Thus, we have
    \begin{eqnarray}\label{diff_r2-r1_Lip}
        \vert R_\freeset(\rho_1)-R_\freeset(\rho_2) \vert \leq (\min(d_A,d_B) - \frac12)\Vert \rho_1-\rho_2\Vert_{\trace} ,
    \end{eqnarray}
    with Lipschitz constant $L = \min(d_A,d_B)-\frac12$.
    \end{proof}

\section{Proof of Theorem \ref{thm:fullrank}} \label{sec:proofthm5}
\setcounter{Thm}{4}
\begin{Thm}
Let $\freeset$ be a subset of $\allset$ that is star-convex w.r.t.~a fixed quantum state $\sigma_0$ of full rank, i.e., with smallest eigenvalue $\lambda>0$. Then the global robustness $R$ is Lipschitz-continuous on its effective domain with Lipschitz constant $L=\frac1\lambda$.
\end{Thm}

    \begin{proof}
    The global robustness can be written as \cite{RoA}:
    \begin{eqnarray}
    R(\rho) = \inf_{\sigma\in \freeset} \{s\geq 0\,:\, \rho \leq (1+s)\sigma\}.
    \end{eqnarray}
    Consider the two quantum states $\rho$ and $\tau$ from the effective domain from $R$ with $\Vert \rho - \tau \Vert_{\trace} = \delta$. This implies that each eigenvalue $\lambda_i$ of $\tau - \rho$,  $\vert\lambda_i(\tau -  \rho)\vert\leq \delta$. Therefore, 
    \begin{eqnarray}
        \tau - \rho \leq \delta \frac{\sigma_0}{\lambda}.
    \end{eqnarray}
    Let $R(\rho) = s^\star$ with optimal state $\sigma^\star \in \freeset$, s.t. $\rho \leq (1+s^\star)\sigma^\star$. Thus, we have the chain of inequalities $\tau - \frac{\delta}{\lambda}\sigma_0 \leq \rho \leq (1+s^\star)\sigma^\star$, from which immediately follows
    \begin{eqnarray}
    \tau &\leq (1+s^\star)\sigma^\star + \frac{\delta}{\lambda}\sigma_0 \\
    &=(1+s^\star + \frac{\delta}{\lambda}) \frac{(1+s^\star)\sigma^\star + \frac{\delta}{\lambda}\sigma_0}{1+s^\star + \frac{\delta}{\lambda}}.
    \end{eqnarray}
    Setting $s^\prime := s^\star + \frac{\delta}{\lambda}$ and $\sigma^\prime :=  \frac{(1+s^\star)\sigma^\star +\frac{\delta}{\lambda}\sigma_0}{1+s^\star + \frac{\delta}{\lambda}}$, we note, using star-convexity of $\freeset$ w.r.t.~$\sigma_0$, that $\sigma^\prime \in \freeset$, as well as
    \begin{eqnarray}
        \tau \leq (1+s^\prime)\sigma^\prime.
    \end{eqnarray}
    While $s^\prime$ and $\sigma^\prime$ might not be optimal points in the set minimized over for $R(\tau)$, the pair definitely is in the set. Thus, $s^\prime$ constitutes an upper bound on $R(\tau)$. Therefore,
    \begin{equation}
       R(\tau) - R(\rho) \leq s^\prime - s^\star =s^\star + \frac{\delta}{\lambda} - s^\star = \frac{\delta}{\lambda}.
    \end{equation}
    Repeating the same argument with switched roles of $\rho$ and $\tau$ yields
    \begin{eqnarray}
        \vert R(\rho) - R(\tau)\vert \leq \frac{1}{\lambda} \Vert \rho - \tau \Vert_{\trace}.
    \end{eqnarray}
    \end{proof}

\section{Proof of Corollary~\ref{crl:fidelity}}

\setcounter{Thm}{5}\label{app:prooffidelity}

\begin{Crl}
    Let $R_\text{tel}$ denote the global robustness of teleportability defined in Eq.~(\ref{eq:rtel}). Then
    \begin{eqnarray}
        \vert R_\text{tel}(\rho) - R_\text{tel}(\tau) \vert \leq (d+1)\Vert \rho - \tau \Vert_{\trace}
    \end{eqnarray}
    for all quantum states $\rho$ and $\tau$.
\end{Crl}

\begin{proof}
The set $\freeset$ contains a ball of states around the maximally mixed state $\one / d^2$ with $\kappa = \frac{d-1}{d^2}$, as can be seen as follows. 
If $\rho \in \mathcal{B}_{\frac{d-1}{d^2}}(\one/d^2)$, then
\begin{eqnarray}
                     & \Vert\rho - \frac{\one}{d^2}\Vert_{\trace} &\leq \frac{d-1}{d^2} \\
     \Leftrightarrow &\max_{X} \frac{\vert \trace[(\rho - \frac{\one}{d^2})X]\vert}{\Vert X \Vert_\infty} &\leq \frac{d-1}{d^2} \\
     \Rightarrow & \max_{U,V \in \text{SU}(d)}  \trace[(\rho - \frac{\one}{d^2})U\otimes V\ketbra{\phi^+}{\phi^+}U^\dagger \otimes V^\dagger] &\leq \frac{d-1}{d^2} \\
     \Rightarrow & \max_{U,V\in \text{SU}(d)} \braket{\phi^+|U^\dagger \otimes V^\dagger \rho U\otimes V|\phi^+} &\leq \frac1d \\
     \Rightarrow & F_\text{max}(\rho) & \leq \frac{1}{d},
\end{eqnarray}
implying that $\rho$ is in the free set.

Defining for this quantity a (global or absolute) robustness of teleportability $R_\text{tel}$, we can use Theorem~\ref{thm:kappaball} to find that it is Lipschitz-continuous with $L=d+1$, and therefore,
\begin{eqnarray}
\vert R_\text{tel}(\rho) - R_\text{tel}(\tau) \vert \leq (d+1)\Vert \rho - \tau \Vert_{\trace}.
\end{eqnarray}

\end{proof}

\section{Proof of Lemma~\ref{lem:belldiagonal}}\label{app:belldiagonal}

\setcounter{Thm}{6}

\begin{Lem}
If $\rho$ is a Bell diagonal state, then $R_{\mathcal{D}_B}(\rho)$ can be computed, via Eq.~(\ref{eq:robdisalt}),  by restricting the infimum to Bell diagonal  free states.
\end{Lem}

\begin{proof}
First, let
\begin{eqnarray}
    \sigma = \frac14(\one \otimes \one + \sum_i x_i \sigma_i \otimes \one + \sum_j y_j \one \otimes \sigma_j + \sum_{ij} t_{ij} \sigma_i\otimes \sigma_j)
\end{eqnarray}
be the Bloch decomposition of a free state such that $\rho \leq (1+s)\sigma$.
It can be shown that $\sigma$ has vanishing discord, iff \cite{dakic2010necessary}
\begin{eqnarray}\label{eq:zerodiscord}
    \Vert \vec{y} \Vert^2 + \Vert T \Vert^2 - k_\text{max} = 0,
\end{eqnarray}
where $\Vert T \Vert = \sqrt{\trace(TT^T)}$ denotes the Frobenius norm of the correlation matrix with entries $t_{ij}$, and $k_\text{max}$ denotes the largest eigenvalue of $\lambda_{\text{max}}(\vec{y} \vec{y}^T + TT^T)$. Note that a value of $0$ implies that $T$ has at most one non-vanishing singular value $k$, hence it can be written as $T=k\vec{u}\vec{v}^T$ with normalized vectors $\vec{u}$ and $\vec{v}$.

From Eq.~(\ref{eq:zerodiscord}), it immediately follows that also the state $\hat{\sigma}$, given by $\sigma$ when setting $\vec{x} = \vec{y} = 0$, is a free state.
In particular, $\hat{\sigma}$ can be written as $\hat{\sigma} = \frac12(\sigma + \tilde{\sigma})$, where $\tilde{\sigma} := (\sigma_2 \otimes \sigma_2)\sigma^T(\sigma_2 \otimes \sigma_2)$ is the universal state inversion of $\sigma$ \cite{rungta2001universal}. As all Bell-diagonal states are invariant under state inversion (which simply puts a minus sign in front of the marginal terms), it follows that if $\rho \leq (1+s)\sigma$, then $\tilde{\rho} = \rho \leq (1+s)\tilde{\sigma}$, as state inversion leaves the eigenvalues invariant. Finally, mixing the two inequalities leads to $\rho \leq (1+s)\frac12(\sigma + \tilde{\sigma}) = (1+s)\hat{\sigma}$.

In summary, we have shown that the optimal state can be chosen to have maximally mixed marginals, as well as a correlation matrix of the form $T=k\vec{u}\vec{v}^T$. It remains to show that there exists a Bell-diagonal free state $\hat{\sigma}_1$ that is free, Bell diagonal and still fulfills $\rho \leq (1+s)\hat{\sigma}_1$.

To that end, we write the optimal state explicitly as $\hat{\sigma} = \frac{1}{4}[\one\otimes \one + \sum_{ij} ku_iv_j \sigma_i \otimes \sigma_j]$. We can find local unitary transformations that diagonalize the correlation matrix $T=k\vec{u}\vec{v}^T$. Thus, we can assume that $\hat{\sigma} =  U\otimes V \hat{\sigma}_{1} U^\dagger \otimes V^\dagger$, where $\hat{\sigma}_1$ has a diagonal correlation matrix with entry $\pm k$ on one of the three diagonal sites, i.e.,
\begin{eqnarray}\label{eq:hatsigma1}
\hat{\sigma}_1 = \frac14(\one \otimes \one \pm k \sigma_a \otimes \sigma_a)    
\end{eqnarray}
for some free choice of $a\in\{1,2,3\}$ and sign of $k$. This state (and therefore $\hat\sigma$ as well) has eigenvalues $\vec{\lambda} = (\frac{1-k}2, \frac{1-k}2, \frac{1+k}2, \frac{1+k}2)$. Note that $\hat{\sigma}_1$ is Bell-diagonal, as well as discord-free, as it is obtained from a discord-free state by local unitary rotations.

Now we will exploit the assumption that $\hat{\sigma}$ is a feasible state for the optimization, i.e., $(1+s)\hat{\sigma} - \rho \geq 0$.
Remember that $\rho$ is Bell-diagonal, i.e., $\rho = \sum_{i=1}^4 p_i \ketbra{\phi_i}{\phi_i}$. Here, the $\ket{\phi_i}$ with $i=1,\ldots,4$ denote the usual Bell states, arranged such that $p_1\leq p_2 \leq p_3 \leq p_4$.
We can now use Weyl's inequalities for the eigenvalues of sums of hermitian matrices \cite{knutson2001honeycombs}. To that end, let us denote the eigenvalues of $(1+s)\hat{\sigma} - \rho$ by $0\leq \kappa_1\leq \kappa_2 \leq \kappa_3 \leq \kappa_4$. Then, Weyl's inequalities yield
\begin{eqnarray}
    \kappa_1 \leq (1+s)\lambda_j - p_j
\end{eqnarray}
for $j=1,\ldots,4$.
As the eigenvalues $\lambda_i = \frac{1\pm k}{4}$ are known and $\kappa_1 \geq 0$, we can rearrange these inequalities into
\begin{eqnarray}\label{eq:weylineqs}
    p_{1,2} \leq \frac{(1+s)(1-k)}4,\quad p_{3,4} \leq \frac{(1+s)(1+k)}4.
\end{eqnarray}

Finally, we can show that the Bell-diagonal $\hat{\sigma}_1$ would constitute a valid choice for the optimization as well. To that end, we exploit that we can write $\hat{\sigma}_1$ as a Bell mixture,
\begin{eqnarray}
    \hat{\sigma}_1 = \frac{k}2(\ketbra{\phi_i}{\phi_i} + \ketbra{\phi_j}{\phi_j}) + \frac{1-k}{4} \one \otimes \one,
\end{eqnarray}
due to the fact that the correlation matrix of the mixture of two Bell states has the form
\begin{eqnarray}
    T[\frac12(\ketbra{\phi^+}{\phi^+} + \ketbra{\phi^-}{\phi^-})]  &= \diag(0,0,1), \nonumber\\
    T[\frac12(\ketbra{\phi^+}{\phi^+} + \ketbra{\psi^+}{\psi^+})]  &= \diag(1,0,0), \nonumber\\
    T[\frac12(\ketbra{\phi^+}{\phi^+} + \ketbra{\psi^-}{\psi^-})]  &= \diag(0,-1,0), \nonumber\\
    T[\frac12(\ketbra{\phi^-}{\phi^-} + \ketbra{\psi^+}{\psi^+})]  &= \diag(0,1,0), \nonumber\\
    T[\frac12(\ketbra{\phi^-}{\phi^-} + \ketbra{\psi^-}{\psi^-})]  &= \diag(-1,0,0), \nonumber\\
    T[\frac12(\ketbra{\psi^+}{\psi^+} + \ketbra{\psi^-}{\psi^-})]  &= \diag(0,0,-1).
\end{eqnarray}

We now choose explicitly the Bell states with largest contribution to $\rho$, i.e., $i=3$, $j=4$ (this corresponds to fixing $a$ and the sign of $k$ in Eq.~(\ref{eq:hatsigma1}). Then, the spectrum of $(1+s)\hat{\sigma}_1 - \rho$ reads $(\frac{(1+s)(1-k)}4 - p_1, \frac{(1+s)(1-k)}4 - p_2, \frac{(1+s)(1+k)}4 - p_3, \frac{(1+s)(1+k)}4 - p_4)$, each of which is positive due to the inequalities  in Eq.~(\ref{eq:weylineqs}).
Thus, the discord-free Bell-diagnonal state $\hat{\sigma}_1$ is a feasible point in the set of the optimization with the same value of $s$. Thus, the optimization can be restricted to such quantum states.
\end{proof}

\section{Proof of Corollary~\ref{crl:rod}}\label{app:proofrod}

\setcounter{Thm}{7}

\begin{Crl}
    Let $\rho = \frac14(\one \otimes \one + \sum_{i=1}^3 x_i \sigma_i \otimes \one + \sum_{i=1}^3 y_j \one \otimes \sigma_j + \sum_{i,j=1}^3 T_{ij}\sigma_i\otimes\sigma_j)$. Then the robustness of discord $R_{\mathcal{D}_B}(\rho)$ is bounded by
    \begin{eqnarray}
        \vert c_2 \vert - 4\max(\vert x\vert, \vert y\vert) \leq R_{\mathcal{D}_B}(\rho) \leq \vert c_2\vert + 4\max(\vert x\vert, \vert y\vert),
    \end{eqnarray}
    where $\vert c_2\vert$ is the singular value of the correlation matrix $T$ with the second largest magnitude and $\vert x\vert$ and $\vert y \vert$ denote the vector norms of the local Bloch vectors of $\rho_A$ and $\rho_B$, respectively.
    
    In particular, if $\rho$ is Bell-diagonal, then
    \begin{eqnarray}
        R_{\mathcal{D}_B}(\rho) = \vert c_2 \vert.
    \end{eqnarray}
\end{Crl}

\begin{proof}

We start with the case of Bell-diagonal states $\rho$. In this case, we know from Lemma~\ref{lem:belldiagonal} that we can perform the minimization in Eq.~(\ref{eq:robdisalt}) by minimizing over Bell-diagonal states only. Thus, we minimize three times, once for each set of free states $\freeset_i = \{\rho_{\text{BDS}}(c_i = c, c_{j\neq i}=0) : -1\leq c \leq 1\}$, and taking the minimum of the three values. Using Lagrange multipliers, this yields after a straight-forward calculation
\begin{eqnarray}
   R_{\mathcal{D}_B}[\rho_{\text{BDS}}(c_1,c_2,c_3)] = \max[\min(\vert c_1\vert, \vert c_2\vert), \min(\vert c_1\vert, \vert c_3\vert), \min(\vert c_2\vert, \vert c_3\vert)].\nonumber
\end{eqnarray}
If we sort the coefficients by increasing absolute value, $\vert c_1\vert \leq \vert c_2\vert \leq \vert c_3\vert$, this simplifies to the middle one, i.e., 
$R_{\mathcal{D}_B}(\rho_{\text{BDS}}(c_1,c_2,c_3)) = \vert c_2\vert$. On the right hand side in Fig.~\ref{fig:rod}, we plot surfaces of constant robustness of discord for this measure.

Now we extend the results to obtain bounds on the robustness of discord for non Bell-diagonal states by using Theorem~\ref{thm:fullrank}.

To that end, we start with a general $\rho$ and notice that the filtered state $\rho_f$ with the marginal terms $x_i$, $y_j$ set to zero and the same correlation terms $T_{ij}$ has robustness $\vert c_2 \vert$, where $\vert c_1\vert \leq \vert c_2\vert \leq \vert c_3\vert$ denote the singular values of $T$. 
On the other hand, direct evaluation of the eigenvalues yields
\begin{eqnarray}\Vert \rho - \rho_f\Vert_{\trace} &= \frac14\Vert \sum_{i=1}^3x_i \sigma_i \otimes \one + \sum_{j=1}^3y_j \one \otimes \sigma_j\Vert_{\trace} \\
&=\max(\vert x\vert, \vert y\vert).
\end{eqnarray} 
As the set of zero discord states is star-convex w.r.t~the maximally mixed state with smallest eigenvalue $\frac14$, the robustness is Lipschitz-continuous with $L=4$ due to Theorem~\ref{thm:fullrank}. Thus, we can bound the robustness of discord of an arbitrary quantum state $\rho$ via
\begin{eqnarray}
    \vert c_2 \vert - 4\max(\vert x\vert, \vert y\vert) \leq R_{\mathcal{D}_B}(\rho) \leq \vert c_2\vert + 4\max(\vert x\vert, \vert y\vert),
\end{eqnarray}
where $\vert c_2\vert$ is the second largest singular value of the correlation matrix $T$ and $\vert x\vert$ and $\vert y \vert$ denote the vector norms of the local Bloch vectors of $\rho_A$ and $\rho_B$, respectively.
\end{proof}

\end{document}